\documentclass[conference,a4paper]{IEEEtran}
\newif\ifcomplete
\completefalse

\usepackage[cmex10]{amsmath} 

\usepackage{graphicx}
\usepackage{epstopdf}
\usepackage{caption}

\usepackage{amsthm, amssymb,accents,mathalfa,bbm}
\usepackage{lipsum}
\usepackage{nccmath}
\usepackage{fancyhdr}
\usepackage{subfigure}
\usepackage{mathdots}
\usepackage{mathtools}
\usepackage{cuted}
\usepackage{cleveref}
\usepackage{algorithm}
\usepackage{algorithmic}
\usepackage{comment}
\usepackage[final,notref,notcite]{showkeys}
\usepackage{enumerate}
\usepackage{cite}
\usepackage[T1]{fontenc}
\usepackage[latin1]{inputenc} 
\usepackage{tikz}
\usetikzlibrary{arrows,snakes,,fit,calc,shapes,backgrounds}
\tikzstyle{box} = [draw, rectangle, rounded corners, thick, node distance=7em, text width=6em, text centered, minimum height=3.5em]
\tikzstyle{container} = [draw, rectangle, dashed, inner sep=2em]\tikzstyle{line} = [draw, thick, -latex']

\DeclareMathOperator*{\argmax}{arg\,max}

\def\TV{\mathrm{TV}}
\def\eqdef{\triangleq}

\relpenalty=\maxdimen

\def\Ber{\mathrm{Ber}}
\def\BP{\mathrm{BP}}

\def\BEC{\mathrm{BEC}}
\def\BSC{\mathrm{BSC}}
\def\BMS{\mathrm{BMS}}

\newtheorem{definition}{Definition}

\newtheorem{remark}{Remark}

\newtheorem{lemma}{\bf Lemma}
\newtheorem{proposition}{\bf Proposition}

\newtheorem{conjecture}{\bf Conjecture}

\mdseries

 \usepackage{cancel}

\def\p{\mathbf{pa}}

\def\P{{\bf P}}

\def\argmax{\mathop{\rm argmax}}

\def\EE{\mathbb{E}}

\def\PP{\mathbb{P}}

\def\eqdef{\triangleq}

\def\BP{\mathrm{BP}}

\def\BEC{\mathrm{BEC}}
\def\BSC{\mathrm{BSC}}
\def\BMS{\mathrm{BMS}}

\def\P{{\bf P}}

\def\eqdef{\triangleq}

\def\al{\alpha}

\IEEEoverridecommandlockouts
\usepackage{atbegshi}
\AtBeginDocument{\AtBeginShipoutNext{\AtBeginShipoutDiscard}}
\begin{document}
\title{Broadcasting on trees near criticality}


\author{%
  \IEEEauthorblockN{Yuzhou Gu, Hajir Roozbehani, and Yury Polyanskiy}
}
\thanks{The authors are affiliated with the Department of EECS and Laboratory of Decision and Information Sciences
(LIDS) at MIT.
Email: \texttt{\{yuzhougu,hajir,yp\}@mit.edu}.
This work was supported in part by the National Science Foundation award under grant agreement CCF-17-17842 and by the Center for Science of Information (CSoI),
an NSF Science and Technology Center, under grant agreement CCF-09-39370.
}%


\maketitle

\begin{abstract}
We revisit the problem of broadcasting on $d$-ary trees:
starting from a Bernoulli$(1/2)$ random variable $X_0$ at a root vertex, each
vertex forwards its value across binary symmetric channels $\BSC_\delta$ to $d$ descendants. The goal is to reconstruct
$X_0$ given the vector $X_{L_h}$ of values of all variables at depth $h$. It is well known that reconstruction (better
than a random guess) is possible as $h\to \infty$ if and only if $\delta < \delta_c(d)$. In this
paper, we study the behavior of the mutual information and the probability of error when $\delta$ is slightly
subcritical. The
innovation of our work is application of the recently introduced ``less-noisy'' channel comparison techniques. For
example, we are able to derive the positive part of the phase transition (reconstructability when $\delta<\delta_c$) using
purely information-theoretic ideas. This is in contrast with previous derivations, which explicitly
analyze distribution of the Hamming weight of $X_{L_h}$ (a so-called Kesten-Stigum bound).
\end{abstract}


\section{Introduction}
We consider the following problem, also known as broadcasting on trees (BOT). Consider an infinite rooted $d$-ary
tree, in which every vertex $v$ has $d$ descendants $v_1,\ldots, v_d$. Let $L_h$ denote all vertices at depth $h$, so
that $|L_h|=d^h$. To each vertex $v$ we associate a binary random
variable $X_v$, whose joint distribution is described inductively as follows. The root variable $X_0 \sim \Ber(1/2)$
is an unbiased Bernoulli. Given all random variables $X_{L_h}$ at depth $h$ the variables at depth $h+1$ are generated
conditionally independently as follows. If $(u,v)$ is an edge in the tree with $u\in L_h$ and $v\in L_{h+1}$ the
(conditioned on $X_{L_h}$) we set $X_v = X_u$ with probability $(1-\delta)$ and $X_{v} =1-X_u$ otherwise. We define the
following quantities\footnote{Throughout this paper, we use $\log$ to denote binary logarithm, and $\ln$ to denote natural logarithm.
Mutual information $I$ is defined with base $2$.}:
\begin{align}
P_e(\delta) &= \lim_{h\to\infty} \PP[X_0 \neq \hat X_0(X_{L_h})],\nonumber\\
	&{} \quad \hat X_0(y_h) = \argmax_{a\in\{0,1\}} \PP[X_0=a|X_{L_h}=y_h]\,,\label{eq:pe_def}\\
   I(\delta) &= \lim_{h\to\infty} I(X_0; X_{L_h})\,.
\end{align}

When $P_e < 1/2$ (equivalently, $I>0$) we say that \textit{reconstruction is possible}. The foundational
work~\cite{BRZ1995} established that the reconstruction is possible if and only if
$$ \delta < \delta_c \eqdef {\frac 1 2} \left(1-{\frac 1 {\sqrt{d}}}\right)\,.$$
We note that the positive part (that $P_e < 1/2$ when $\delta < \delta_c$) follows from a so-called Kesten-Stigum bound,
cf.~\cite{evans2000broadcasting}, which in fact proves that reconstruction can be done by a sub-optimal detector
\begin{equation}\label{eq:ks}
		\hat X_{0,maj}(y_h) = 1\{\|y\|_H > d^{h-1}/2\}\,,
\end{equation}
	where $\|y\|_H = |\{j: y_j \neq 0|$ is the Hamming weight.
The extension to general (non-regular) trees was done in~\cite{evans2000broadcasting} and beyond trees
in~\cite{makur2019broadcasting}. There are deep connections between BOT and other problems. In statistical physics it
arises in the study of the free-boundary Gibbs measure for the Ising model on a tree~\cite{BRZ1995},
problems on random graphs~\cite{mezard2006reconstruction} and in random constraint
satisfaction~\cite{montanari2011reconstruction}. It was a key step for establishing the sharp thresholds for
the problem of community detection (stochastic block model)~\cite{mossel2014belief}. It can be seen as a simple model
for genetic mutations~\cite{mossel2003impossibility} and noisy computations~\cite{evans1993signal}.

We note that various theories (starting from Ginzburg-Landau)
in statistical physics predict the type of behavior of various quantities in the vicinity of the phase transition (the
so-called \textit{critical exponents}).
However, to the best of our knowledge, the behavior of $I(\delta)$ and $P_e(\delta)$ near the critical point $\delta =
\delta_c - \tau$ with $\tau \ll 1$ is not understood. In particular, the best results available in the literature show
that for some $0<c_1<c_2$ and $c_3,c_4>0$ we have
\begin{align}\label{eq:i_bd}
		c_1 \tau + o(\tau) &\le I(\delta_c - \tau) \le c_2 \tau + o(\tau),\\
\label{eq:pe_bd}
		\tfrac12 - c_3 \sqrt{\tau} + o(\sqrt{\tau}) &\le P_e(\delta_c-\tau) \le \tfrac12 - c_4 \tau + o(\tau)\,.
\end{align}
The main open question is establishing the critical exponent in~\eqref{eq:pe_bd}. This can be phrased
equivalently as follows: The upper bound in~\eqref{eq:pe_bd} can only be tight if in the regime $\delta \uparrow
\delta_c$ the induced (BMS) channel $X_0 \mapsto X_{L_h}$ resembles an erasure channel, whereas the lower bound can only be
tight if the channel resembles ``more diffuse'' BMS channel, akin to binary-input AWGN one. Thus, settling the exponent
in~\eqref{eq:pe_bd} can be rephrased as the question of understanding the kind of residual uncertainty about $X_0$
remaining after observing the leaves $X_{L_h}$.

Our contributions are as follows:
\begin{enumerate}
\item We adapt the channel comparison technique from~\cite{roozbehani2019low} to the study of $I$ and $P_e$ and show in
particular that the positive part can be established without analyzing either the suboptimal decoder~\eqref{eq:ks} (as
in Kesten-Stigum) or the belief propagation (as in~\cite{BRZ1995}).
\item We improve available estimates on $c_2$.
\item We develop a sequence of numerically computable bounds, each provably upper and lower bounding $I$ and $P_e$,
whose evaluation allows us to make the following two conjectures for binary trees.
\end{enumerate}
\begin{conjecture}
$I(\delta_c - \tau) = \frac{4\sqrt{2}}{\ln 2}\tau+o(\tau)$ bit.
\end{conjecture}

\begin{conjecture}
$P_e(\delta_c - \tau)=1/2-\Theta(\sqrt{\tau}).$
\end{conjecture}

\subsection{Channel comparison lemmas}
We quickly review the channel comparison lemmas of \cite{roozbehani2019low} and discuss how they relate to broadcasting. We start with reviewing some key information-theoretic notions.

\begin{definition}[{~\cite[\S 5.6]{el2011network}}]
Given two channels $P_{Y|X}$ and $P_{Y\sp{\prime}|X}$ with common input alphabet, we say that $P_{Y\sp{\prime}|X}$ is
\begin{itemize}
    \item {\em less noisy} than $P_{Y|X}$, denoted by $P_{Y|X} \preceq_\mathrm{l.n.} P_{Y\sp{\prime}
|X}$, if for all joint distributions $P_{UX}$ we have
\[
I(U;Y)\le I(U;Y')
\]
\item {\em more capable} than $P_{Y|X}$, denoted by $P_{Y|X}\preceq_{\mathrm{m.c.}} P_{Y\sp{\prime}
|X}$, if for all marginal distributions $P_{X}$ we
have
$$I(X; Y ) \le I(X; Y\sp{\prime}).$$
\item {\em less degraded} than $P_{Y|X}$, denoted by $P_{Y|X}\preceq_{\mathrm{deg}}P_{Y\sp{\prime}|X}$, if there exists a Markov chain $Y-Y'-X$.
\end{itemize}
\end{definition}
We refer to
~\cite[Sections I.B, II.A]{makur2018comparison} and~\cite[Section 6]{polyanskiy2017strong} for alternative useful characterizations of the less-noisy order.

For an arbitrary pair of random variables we define
$$ I_{\chi^2}(X;Y) = \chi^2(P_{X,Y} \| P_X \otimes P_Y)\,,$$
where $P_X \otimes P_Y$ denotes the joint distribution on $(X,Y)$ under which they are independent.

Let $W$ be a $\BMS$ channel  (cf. \cite{roozbehani2019low}, Definition 7), $X\sim \Ber(1/2)$ and $Y=W(X)$ be the output induced by $X$. We
define $W$'s probability of error, capacity, and $\chi^2$-capacity as follows
	\begin{align}
		P_e(W) &= \frac{1-\TV(W(\cdot|0), W(\cdot|1))}{ 2},\\
	   C(W) &= I(X;Y),\\
	   C_{\chi^2}(W) &= I_{\chi^2}(X;Y)\,.
	\end{align}

\begin{lemma}{~\cite[Lemma 2]{roozbehani2019low}}\label{lem:lemma_A}
  The following holds:
  \begin{enumerate}
  \item Among all $\BMS$ channels with the same value of $P_e(W)$ the least
  degraded is $\BEC$ and the most degraded is $\BSC$, i.e.
\begin{equation}\label{eq:la_deg}
    	\BSC_{\delta} \preceq_{deg} W \preceq_{deg} \BEC_{2\delta}\,,
\end{equation}
  where $\preceq_{deg}$ denotes the (output) degradation order.

  \item Among all $\BMS$ with the same capacity $C$ the most capable is $\BEC$ and the least capable is
  $\BSC$, i.e.:
\begin{equation}\label{eq:la_mc}
    	\BSC_{1-h_b^{-1}(C)} \preceq_{mc} W \preceq_{mc}  \BEC_{1-C}\,,
\end{equation}
  where $\preceq_{mc}$ denotes the more-capable order, and $h_b^{-1}:[0,1]\to[0,1/2]$ is the functional inverse of the (base-2) binary entropy function $h_b:[0,1/2]\to[0,1]$.

  \item Among all $\BMS$ channels with the same value of $\chi^2$-capacity $\eta=I_{\chi^2}(W)$
  the least noisy is $\BEC$ and the most noisy is $\BSC$, i.e.
\begin{equation}\label{eq:la_ln}
  	\BSC_{1/2-\sqrt{\eta}/2} \preceq_{ln} W \preceq_{ln} \BEC_{1-\eta}\,,
\end{equation}
  where $\preceq_{ln}$ denotes the less-noisy order.
  \end{enumerate}
\end{lemma}

The next lemma states that if the incoming messages to BP are comparable, then the output messages are comparable as well.
\begin{lemma}{~\cite[Lemma 3]{roozbehani2019low}}\label{lem:lemma_B}
Fix some random transformation $P_{Y|X_0,X_1^m}$ and $m$
  $\BMS$ channels $W_1,...,W_m$. Let $W: X_0\mapsto (Y,Y_1^m)$ be a (possibly non-$\BMS$) channel
  defined as follows. First, $X_1,..., X_m$ are generated as i.i.d $\Ber(1/2)$. Second,
  each $Y_j$ is generated as an observation of $X_j$ over the $W_j$, i.e. $Y_j=W_j(X_j)$ (observations are all
  conditionally independent given $X_1^m$). Finally, $Y$
  is generated from $X_0,X_1^m$ via $P_{Y|X_0,X_1^m}$ (conditionally independent of $Y_1^m$ given $X_1^m$). Define the $\tilde W$ channel similarly, but with $W_j$'s replaced with $\tilde W_j$'s. The following statements hold:
  \begin{enumerate}
   \item If $\tilde W_j \preceq_{deg} W_j$ then $\tilde W \preceq_{deg} W$.
   \item If $\tilde W_j \preceq_{ln} W_j$ then $\tilde W \preceq_{ln} W$.
   \end{enumerate}
\end{lemma}
\begin{remark}
An analogous statement for more capable channels does not hold (see  Example 2 in \cite{roozbehani2019low}).
\end{remark}
\begin{definition}[Erasure function]
Consider a single layer of a d-ary tree with source $X_0$. Suppose that each boundary node is observed through a (memoryless) $\BEC$ channel, i.e., $Y^{(j)}=\BEC_{q}(X^{(j)})$ where $q$ is the probability of erasure. The function $$E^{\BEC}(q)\triangleq\EE [\P(X_0=1|Y^{(1)},\cdots,Y^{(d)})|X_0=0].$$ is called the erasure function of the tree. Here the expectation is taken with respect to the  randomization over bits as well as the noise in the observations.
\label{def:erasure_function}
\end{definition}
\begin{definition}[Error function]
In the setup of Definition \ref{def:erasure_function}, let $Y^{(j)}=\BSC_{q}(X^{(j)})$ where $q$ is the crossover probability. The function $$E^{\BSC}(q)\triangleq\EE [\P(X_0=1|Y^{(1)},\cdots,Y^{(d)})|X_0=0].$$ is called the error function of the tree. Here the expectation is taken with respect to the  randomization over bits as well as the noise in the observations.
\label{def:error_function}
\end{definition}
\begin{definition}[$\chi^2$-entropies]
Take the setup of Definition \ref{def:erasure_function}. Let $Y_i$'s be $\BEC$ induced observations as before. Define the  erasure $\chi^2$-entropy function to be
\[
\mathcal{H}^\BEC(q)\triangleq\EE[1-I_{\chi^2}(X_0;Y^{(1)},\cdots,Y^{(d)})].
\]
The corresponding error $\chi^2$-entropy $\mathcal{H}^\BSC$ is defined in an analogous manner to Definition \ref{def:error_function}.
\label{def:H_chi_func}
\end{definition}

The next proposition shows that the broadcasting problem can be cast into the setting of comparison lemmas.
\begin{proposition}
Consider a single layer of a d-ary tree with source $X_0$ and independent observations $X_i=\BSC_\delta(X_0)$ along the edges. Consider the channels $W:X_0\mapsto Y^d_1$ with $Y_i=W_i(X_i)$ and $\tilde W:X_0\mapsto \tilde Y_1^d$ with $\tilde Y_i=\tilde {W}_i (X_i)$. The following statements hold:
  \begin{enumerate}
   \item If $\tilde W_j \preceq_{deg} W_j$ then $\tilde W \preceq_{deg} W$.
   \item If $\tilde W_j \preceq_{ln} W_j$ then $\tilde W \preceq_{ln} W$.
   \end{enumerate}
\label{prop:comp_broadcast}
\end{proposition}
\begin{proof}
Let $X'\sim\Ber(1/2)^{\otimes d}$. Define the parity codes $Y'_i=X_0+X_i'$. Note that the channel $X_0\to X_i$ is equivalent to $X_0\to (Y'_i,\BSC_\delta(X'_i))$. Likewise, the channels $W_i$ are equivalent to $X_0\to (Y'_i,W_i(\BSC_\delta(X'_i)))$. This latter map is of the form in Lemma \ref{lem:lemma_B}, from which both statements follow.
\end{proof}
As a consequence we have the following propositions for the broadcasting problem.
\begin{proposition}
Consider the dynamical systems
\begin{align}
q^\BEC_{t+1}(x)&=2E^\BEC(q^\BEC_{t}(x)),
\label{eq:qbec}\\
q^\BSC_{t+1}(x)&=E^\BSC(q^\BSC_{t}(x)),
\label{eq:qbsc}
\end{align}
initialized at $q^\BEC_0(x)=q^\BSC_0(x)=x$.
Let $P_e(\mathcal{T}_\ell)$ be the probability of error under BP after broadcasting on a $d$-ary tree of depth $\ell$. Then
\[
\frac{q^\BEC_\ell(0)}{2}\le P_e(\mathcal{T}_\ell)\le
 q^\BSC_\ell(0).
\]
\label{prop:approx_gives_lower_bound}
\end{proposition}
\begin{proof}
The proof follows from that of \cite[Proposition 9]{roozbehani2019low} upon replacing Lemma \ref{lem:lemma_B} with Proposition \ref{prop:comp_broadcast}.
\end{proof}
\begin{proposition}
Consider the dynamical systems
\begin{align}
q^\BEC_{t+1}(x)&=\mathcal{H}^\BEC(q^\BEC_{t}(x)),
\label{eq:qbec_chisoft}\\
q^\BSC_{t+1}(x)&=1/2-1/2\sqrt{1-\mathcal{H}^\BSC(q^\BEC_{t}(x))},
\label{eq:qbsc_chisoft}
\end{align}
initialized at $q^\BEC_0(x)=q^\BSC_0(x)=x$. 
Let $I(X_0;\mathcal{T}_\ell)$ be the mutual information between root and observed leaves at depth $\ell$. Then
\[
1-q^\BEC_\ell(0)\ge I(X_0;\mathcal{T}_\ell)\ge 1-h(q^\BSC_\ell(0)),
\]
where $h$ is the binary entropy function.
\label{prop:approx_gives_lower_bound_info}
\end{proposition}
\begin{proof}
The proof follows easily from that of \cite[Proposition 9]{roozbehani2019low} upon replacing Lemma \ref{lem:lemma_B} with Proposition \ref{prop:comp_broadcast}.
\end{proof}

\section{The reconstruction threshold}
In this section we prove the reconstruction threshold using the channel comparison method.

\def\ep{\epsilon}
\def\de{\delta}
\def\lm{\lambda}
\def\ka{\kappa}
\def\bP{\mathbb P}
\def\p{\prime}
\def\Th{\Theta}

\begin{proposition}\label{prop:ks_d_ary_lower_bound}
If $d(1-2\delta)^2>1$, then recovery (better than random guess) is possible on $d$-ary trees.
\end{proposition}
\begin{proof}
  By Proposition \ref{prop:approx_gives_lower_bound_info}, it suffices to show that the $\chi^2$-dynamics for $\BSC$ expand the information in a neighborhood of $0$.
  Consider a $d$-ary tree with source $X_0$.
  Suppose that its children $X_1,\ldots, X_d$ are observed with some probability $\lambda$ through a BSC channel. Let $Y_1,\ldots,Y_d$ be the observations.
  Because we work in a neighborhood of $0$, we write $\lm = \frac 12 - \ep$ with $\ep > 0$ very small.
  For simplicity, write $\ka := \de * \lm = \frac 12 - (1-2\de)\ep$.
  Then by definition we have
  \begin{align*}
    &I_{\chi^2} (X_0; Y)\\
    &= \sum_{x_0\in \{0, 1\}} \sum_{y\in \{0, 1\}^d} \frac{\bP(X_0=x_0, Y=y)^2}{\bP(X_0=x_0)\bP(Y=y)} - 1\\
    & = 2 \sum_{0\le i\le d} \binom d i \frac{\ka^{2i} (1-\ka)^{2(d-i)}}{\ka^i (1-\ka)^{d-i} + \ka^{d-i} (1-\ka)^i}-1.
  \end{align*}
  Using the formula
  \begin{align*}
    &\ka^a (1-\ka)^b + \ka^b (1-\ka)^a \\
    &= 2^{1-a-b} (1 + (\binom a2 + \binom{b}2 - ab) 4(1-2\de)^2\ep^2 + O(\ep^4)),
  \end{align*}
  we can expand in terms of $\ep$ and get
  \begin{align*}
    &I_{\chi^2} (X_0; Y) \\
    & =
    \sum_{0\le i\le d} \binom di 2^{-d} (1+
    (\binom{2i}2+\binom{2(d-i)}2-4i(d-i)\\
    & - \binom i2 - \binom {d-i}2+i(d-i)) 4(1-2\de)^2\ep^2)
    + O(\ep^4)-1\\
    & = 4 d (1-2\de)^2 \ep^2 + O(\ep^4).
  \end{align*}
  Note that the input $\chi^2$-information into the local neighborhood is $4\epsilon^2$ under our parametrization. Thus
denoting by $I_{\chi^2}^t$ the amount of $\chi^2$-information between a target node and its leaves left after $t$ iterations, we get
\[
I_{\chi^2}^t=d(1-2\delta)^2I_{\chi^2}^{t-1}(1+o(1))
\]
This means that if $d(1-2\delta)^2>1$, then for small enough $\epsilon$ the dynamics expand the information and hence the input information cannot contract to $0$ no matter how small it is.
\end{proof}

Likewise, BEC comparisons recover the following result:
\begin{proposition}\label{prop:ks_d_ary_upper_bound}
If $d(1-2\delta)^2 \le 1$ and $(d,\de)\ne (1, 0)$, then recovery (better than random guess) is impossible on $d$-ary trees.
\end{proposition}
\begin{proof}
By Proposition \ref{prop:approx_gives_lower_bound_info}, we need to show that $\BEC$ dynamics contracts information.
Let $X_1, \ldots, X_d$ be the children of $X_0$ and $Y_1,\ldots, Y_d$ be their observations through a $\BEC_{1-\ep}$ channel
Applying Lemma \ref{lem:lemma_B} to the composed channel $X\to X_i\to Y_i$, we see that we can replace $Y_i$ with $Y_i^\p$,
where each $X\to Y_i^\p$ is an independent copy of $\BEC_{1-(1-2\de)^2\ep}$.
We have
\begin{align*}
  I_{\chi^2}(X_0;Y^\p) = 1-(1 - (1-2\de)^2 \ep)^d.
\end{align*}
The input information is $\ep$ under our parametrization.
Consider the function $f(\ep) = 1-(1-(1-2\de)^2\ep)^d$.
We have $f(0)=0$ and $$f^\p(\ep) = d (1-2\de)^2 (1-(1-2\de)^2\ep)^{d-1}.$$
So $f^\p(\ep) \le 1$ for $\ep\in [0, 1]$, and equality is only achieved at $\ep=0$.
So $f$ has only one fixed point in $[0, 1]$, which is $0$.
Therefore $\chi^2$-information contracts to $0$.
\end{proof}

\begin{remark}
  In the proof of Proposition \ref{prop:ks_d_ary_lower_bound}, we showed that when the input information is close to $0$, in the limit the information would contract to a non-zero value. Therefore our proof in fact shows that robust reconstruction (a stronger condition than reconstruction) on such trees is possible.
  By \cite{janson2004robust}, for broadcasting on trees, the robust reconstruction threshold coincides with the Kesten-Stigum bound.
  It is shown in \cite{sly2009reconstruction} that when the alphabet size is at least five, the Kesten-Stigum bound is never tight for the (non-robust) reconstruction problem.
  So for large alphabet size, our method does not yield tight reconstruction threshold.
\end{remark}

\section{Bounds on mutual information}
\begin{proposition}\label{prop:linear_growth}
  Let $d\ge 2$ and $\de = \de_c - \tau$ where $d(1-2\de_c)^2=1$. Let $T_\ell$ be the $d$-ary tree channel as in above.
  Then
  \begin{align*}
    \frac{2 d \sqrt d}{(d-1)\ln 2} \tau+o(\tau) &\le \lim_\ell I(X_0; T_\ell)\\
    &\le \frac{4(d+1)\sqrt d}{d-1} \tau+o(\tau).
  \end{align*}
\end{proposition}
\begin{proof}
  The proof is by analyzing the recursion in the proof of Proposition \ref{prop:ks_d_ary_lower_bound} and \ref{prop:ks_d_ary_upper_bound} more carefully.

  In the setting of proof of Proposition \ref{prop:ks_d_ary_lower_bound}, expanding everything to the order of $\ep^4$ and computing a binomial sum, we get
  \begin{align*}
    &I_{\chi^2}(X_0; Y) \\
    &= 4 d (1-2\de)^2 \ep^2 + 16 d(d-1) (1-2\de)^4 \ep^4 + O(\ep^6)\\
    & = 4 (1+4 \sqrt d \tau + o_\tau(\tau)) \ep^2 + 16 (\frac{d-1}d + o_\tau(1)) \ep^4 + O(\ep^6).
  \end{align*}
  The input information is $4\ep^2$ under this parametrization.
  Solving the dynamics, we get $$\ep^* = (\sqrt{\frac{d\sqrt d}{d-1}} +o(1))\sqrt \tau.$$
  This gives
  \begin{align*}
    \lim_\ell I(X_0; T_\ell) &\ge 1 - h(\frac 12 - (\sqrt{\frac{d\sqrt d}{d-1}} +o(1))\sqrt \tau) \\
    & = \frac{2 d \sqrt d}{(d-1)\ln 2} \tau + o(\tau).
  \end{align*}

  Following the proof of proof of Proposition \ref{prop:ks_d_ary_upper_bound}, let us consider
  the function $f(\ep) = 1-(1-(1-2\de)^2\ep)^d$.
  Now the function $f(\ep)$ is concave on $[0, 1]$, and there is a unique fixed point in $(0, 1)$.
  By expanding in terms of $\ep$, we have
  \begin{align*}
    f(\ep) & = d (1-2\de)^2 \ep - \binom d2 (1-2\de)^4 \ep^2 + O(\ep^3)\\
    & = (1+4 \sqrt d \tau + o_\tau(\tau))\ep - (\frac{d-1}{2d}+o_\tau(1)) \ep^2 + O(\ep^3).
  \end{align*}
  So the unique fixed point is at $$\ep^*=\frac{8d\sqrt d}{d-1} \tau + o(\tau).$$
  This gives
  $$\lim_\ell I(X_0; T_\ell) \le \frac{8d\sqrt d}{d-1}\tau + o(\tau).$$

  In fact, knowing that the limit is linear in $\tau$, we can improve this upper bound.
  Instead of considering $I(X; Y^\p)$ in the proof of Proposition \ref{prop:ks_d_ary_upper_bound},
  let us consider $I(X; Y)$ directly.
  We can compute that
  \begin{align*}
  I_{\chi^2}(X_0;Y)
  & = \sum_{x_0\in \{0, 1\}} \sum_{y\in \{0, 1,*\}^d} \frac{\bP(X_0=x_0, Y=y)^2}{\bP(X_0=x_0)\bP(Y=y)} - 1\\
  & = 2 \sum_{0\le j\le i\le d} \binom di \ep^i (1-\ep)^{d-i} \binom ij\\
  & \cdot \frac{(1-\de)^{2j} \de^{2(i-j)}}{(1-\de)^j \de^{i-j} + (1-\de)^{i-j} \de^j} - 1.
  \end{align*}
  Let us call this function $g(\ep)$.
  Note that by Lemma \ref{lem:lemma_B}, we always have $g(\ep) \le f(\ep)$ on $[0, 1]$.
  So the largest fixed point of $g$ is upper bounded by the non-trivial fixed point of $f$, which is of order $\Th(\tau)$.
  This justifies performing series expansion in $\ep$.
  \begin{align*}
  g(\ep)& =  (1-\ep)^d + 2d ((1-\de)^2+\de^2) \ep(1-\ep)^{d-1}\\
  & + d(d-1) (\frac{(1-\de)^4+\de^4}{(1-\de)^2+\de^2} + (1-\de)\de) \ep^2(1-\ep)^{d-2} \\
  & + O(\ep^3)-1\\
  & = d (1-2\de)^2 \ep - d(d-1) \frac{(1-2\de)^4}{(1-2\de)^2+1} \ep^2 + O(\ep^3)\\
  & = (1+4 \sqrt d \tau + o_\tau(\tau))\ep - (\frac{d-1}{d+1}+o_\tau(1)) \ep^2 + O(\ep^3).
  \end{align*}
  We see that the largest fixed point of $g$ must satisfy
  $$\ep^*=\frac{4(d+1)\sqrt d}{d-1} \tau + o(\tau).$$
  In this way we get
  $$\lim_\ell I(X_0; T_\ell) \le \frac{4(d+1)\sqrt d}{d-1}\tau + o(\tau).$$
\end{proof}

\begin{remark}
We compare the above lower bound with (7) in \cite{evans2000broadcasting}.\footnote{\cite{evans2000broadcasting} contains an error stating that $I \ge I_{\chi^2}$, which should be $I\ge \frac 12 I_{\chi^2}$. (Note that they define mutual information with natural logarithm.) This leads to lower bounds on $I$ (e.g., (4)(28) in \cite{evans2000broadcasting}) to be off by a factor of $2$. (7) in \cite{evans2000broadcasting} is correct as stated.} We note that the lower bound of \cite{evans2000broadcasting} can in the limit be simplified into
\[
\lim_{\ell\to \infty} I_{\chi^2}(X_0;T_\ell)\ge \frac 1{1+\frac {1-(1-2\de)^2}{d (1-2\de)^2-1}}.
\]
Near the critical threshold, RHS behaves as $\frac {4d\sqrt d}{d-1}\tau$.
So they obtained the the same $\chi^2$-information lower bound, thus the same mutual information lower bound,
as in Proposition \ref{prop:linear_growth}.

\cite{evans2000broadcasting} did not state explicitly an upper bound on mutual information. Nonetheless, their upper bound is by comparison with percolation, and that leads to an upper bound of
$$\lim_{\ell\to \infty} I(X_0; T_\ell) \le \frac{8d\sqrt d}{d-1}\tau + o(\tau).$$
In this case we see that channel comparison leads to a better upper bound.
\end{remark}

In the case of binary trees, we perform a more refined analysis to improve the upper bound.
\begin{proposition}
Let $\delta=\delta_c-\tau$ with $2(1-2\delta_c)^2=1$. Let $T_\ell$ be the binary tree channel as in above. Then
\[
  \lim_\ell I(X_0;T_\ell)\le 8(\sqrt 2+1) (1-h(\frac 12 - \sqrt{\frac 1{\sqrt 2}-\frac 12}))\tau+o(\tau).
\]
\label{prop:linear_growth_improved}
\end{proposition}
\begin{proof}
  Suppose the input distribution is a mixture of $\BSC_\Delta$ for $\Delta$ supported at $\{1/2-\al_t, 1/2\}$.
  We iterate the dynamics of Proposition \ref{prop:approx_gives_lower_bound_info}
  while finding the best (w.r.t the less noisy order) channel within this family.
  This family contains $\BEC$ (corresponding to $\al=1/2$), so this approach may lead to a better bound.
  We define \[\bar{\delta}:=(1/2-\alpha)*\delta=1/2-\alpha(1-2\delta).\]
  The output distribution has support $\{\frac{\bar{\delta}^2}{\bar{\delta}^2+(1-\bar{\delta})^2)},\bar{\delta},1/2\}$.
  Using Lemma \ref{lem:lemma_A}, we replace $\BSC_{\bar{\de}}$ with a mixture of $\BSC_{1/2}$ and $\BSC_{\frac{\bar{\delta}^2}{\bar{\delta}^2+(1-\bar{\delta})^2)}}$, while preserving $\chi^2$-information.
  Therefore $$1/2-\al_{t+1} = \frac{\bar{\delta}^2}{\bar{\delta}^2+(1-\bar{\delta})^2}.$$
  Solving this, we get that in the $\ell$ limit $$\al^* = \frac{\sqrt{1-4\de}}{2(1-2\de)}.$$
  For $\al=\al^*$, we have $$C_{\chi^2}(\BSC_{\bar{\delta}}) = (1-2\de)^2 C_{\chi^2}(\BSC_{\frac{\bar{\delta}^2}{\bar{\delta}^2+(1-\bar{\delta})^2}}).$$
  So when applying Lemma \ref{lem:lemma_A}, every unit weight for the former becomes $(1-2\de)^2$ weight for the latter.

  Let $\ep_t$ be the weight of $\BSC_{1/2}$ in iteration $t$.
  Then in the $\ell$ limit $\ep$ should satisfy
  $$1-\ep = (1-\ep)^2(\bar{\delta}^2+(1-\bar{\delta})^2) + 2\ep(1-\ep) (1-2\de)^2.$$
  Solving this we get $\ep^* = 1-8(\sqrt 2+1)\tau + o(\tau)$.

  So an upper bound for mutual information is
  \begin{align*}
    &(1-\ep^*)(1-h(1/2-\al^*)) \\
    &= 8(\sqrt 2+1) (1-h(\frac 12 - \sqrt{\frac 1{\sqrt 2}-\frac 12}))\tau+o(\tau).
  \end{align*}
\end{proof}
\begin{remark}
The same method can be applied to the lower bound, leading to $\al_* = (\sqrt{3\sqrt{2}}+o(1))\sqrt\tau$ and $\ep_* = \frac 13+o(1)$,
giving $$\lim_{\ell} I(X_0; T_\ell)\ge \frac{4\sqrt 2}{\ln 2} \tau + o(\tau).$$
Surprisingly, although we lower bound using a larger family, and the limiting distribution is different, we get the same lower bound as Proposition \ref{prop:linear_growth}.

We have shown that $I(X_0;T_\ell)=c\tau+o(\tau)$ for some $c\in [8.16,14.21]$. The improvement over Proposition \ref{prop:linear_growth} can be attributed to a finer ``quantization'' since we try to work with less noisy channels while staying closer to the true output of BP. We shall explore this idea further in Section \ref{sec:local_comp} and show (numerically) that the correct slope is $c\approx 8.16$.
\end{remark}
\section{Improved bounds via local comparisons}
\label{sec:local_comp}
One advantage of the comparison method is that it allows us to analyze BP, rather than some suboptimal algorithm. On the other hand, we incur some loss in each step of the analysis due to the crude approximations that are made to the input distribution in order to simplify the analysis. In some cases these losses can be significant. For instance, a  naive application of the comparison method while matching probabilities of error (i.e., using least degraded channels and Proposition \ref{prop:approx_gives_lower_bound}) does not even recover the right threshold. One way to avoid this issue is to do local comparisons. We first define a few quantizing operators.
\begin{definition}[Q-Operators]
Consider a binary random variable $X$ with probability law $\mu$ along with quantization intervals $(a_i,b_i)$. Define the quantized BSC operator $Q^\BSC(X)$ as follows:  replace the support of $\mu$ along each $(a_i,b_i)$ with a single point at $\delta_i:=\frac{\int_{a_i}^{b_i} \delta d\mu}{\int_{a_i}^{b_i}  d\mu}$ with probability mass $\int_{a_i}^{b_i}  d\mu$. Likewise, define the quantized BEC operator $Q^\BEC(X)$ as follows: replace the support along $(a_i, b_i)$ with two quantization points $a_i, b_i$ with  probabilities $p_{a_i}:=\alpha_i p_i$, $p_{b_i} = (1-\alpha_i)p_i$, where  $p_i=\int_{a_i}^{b_i} d\mu$ and $\alpha_i=\frac{b_i-\int_{a_i}^{b_i} \delta d\mu/\int_{a_i}^{b_i}  d\mu}{b_i-a_i}$. Furthermore, define $Q^\BSC_{\chi^2}$ (resp. $Q^\BEC_{\chi^2}$) similarly by matching the $\chi^2$-information along each interval while contracting (resp. spreading) probability masses.
\end{definition}

The main idea is presented in the next proposition:
\begin{proposition}
Consider broadcasting on a tree with parameter $\delta$. Suppose that $\mu_0^\BSC$ (the law at the boundary of the tree) is induced by $\BSC_{\delta_0}$, where $\delta_0$ is  chosen so that $\delta_0\ge \lim_\ell P_e(\mathcal{T}_\ell)$. 
 Let $\mu^\BSC_t=Q^\BSC(\mathrm{BP}(\mu^\BSC_{t-1}))$ be obtained by quantizing the output of $\BP$ operating on $\mu^\BSC_{t-1}$. Let $q_t^\BSC$ be the corresponding probability of error $q^\BSC_t:=\sum_i \delta_{a_ib_i}\mu_{t,i}^\BSC$. Similarly, define $\nu^\BSC_t=Q^\BSC_{\chi^2}(\nu^\BSC_{t-1})$. Let $\iota^\BSC_t$ be the corresponding mutual information. Likewise, define $\mu^\BEC_t=Q^\BEC(\BP(\mu^\BEC_{t-1}))$ with probability of error $q^\BEC_t$ with $P(\mu_0=0)=1$. Define $\nu^\BEC_t,\iota^\BEC_t$ similarly. The following statements hold:
\begin{enumerate}
\item $q_\ell^\BEC\le \lim_\ell P_e(\mathcal{T}_\ell)\le q_\ell^\BSC.$
\item $\iota_\ell^\BEC\ge \lim_\ell I(\mathcal{T}_\ell)\ge \iota_\ell^\BSC.$
\end{enumerate}
\end{proposition}
\begin{remark} To choose $\delta_0<\tfrac12$ we may, for example, use a Kesten-Stigum upper bound on $P_e$, corresponding
to a suboptimal algorithm as in~\eqref{eq:ks}.
\end{remark}
\begin{proof}
Note that $Q^\BSC(\mu)$ is obtained from $\mu$ by the transformation \[
(y,\delta)\mapsto \left\{\begin{array}{cc}(y,\delta_i)& \delta \in [a_i, b_i], \\ (y,\delta)& \text{o.w.}\\ \end{array}\right.
\]
The probabilities of error match by construction. This shows that $Q^\BSC(\mu_t)$ is a degradation of $\mu_t$. This proves the upper bound since if the initial input is degraded w.r.t $\mathcal{T}_\ell$ then all the subsequent iterations remain degraded. For the lower bound, we note that a probability distribution with its masses at the center of an interval is a degradation of one with two spikes at the boundaries. Note that indeed when the original distribution has a single atom in some interval this follows directly from the above transformation. The general case follows since if there are more than one atoms, we can degrade sequentially. This proves the first statement. The second statement can be proved similarly.
\end{proof}
Using uniform quantization in the $[0,1/2]$ interval with $1024$ points, we were able to show that
\[
I(X_0;\mathcal{T}_\ell)=c\tau+o(\tau)
\]
with $c\approx 8.16$.
This is the basis of our Conjecture 1.

Using a degradation argument (or Fano's inequality), one can also show
\[
1/2-c'\sqrt{\tau}+o(\sqrt{\tau})\le P_e(\mathcal{T}_\ell)\le 1/2-c\tau+o(\tau).
\]
It is natural to ask what is the correct exponent for $P_e$. Using the same approach we were able to show (see Fig. 1)
\[
\log(1-2P_e)\ge 0.504\log \tau+c.
\]
We thus conjecture that  $\sqrt{\tau}$ is the correct exponent.


\begin{figure}[t]
\centering
\includegraphics[width=0.5\textwidth]{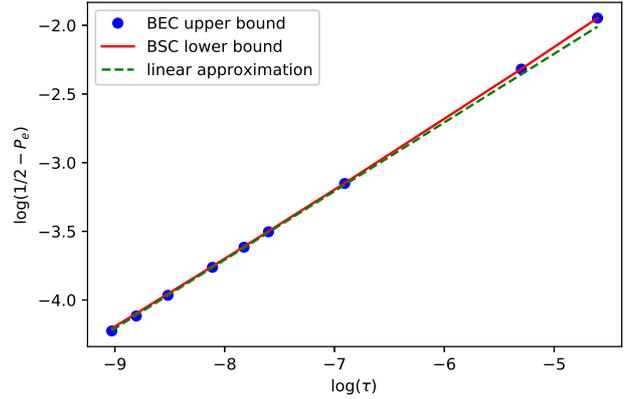}
\caption{Bounds on probability of error using local comparisons for $\delta=\delta_c-\tau$. The linear approximation has a slope of 1/2.}
\label{fig:bound_pe_local}
\end{figure}

\newpage

\bibliographystyle{IEEEtran}

\bibliography{biblio}

\end{document}